\documentclass[a4paper, 12pt, reqno]{amsart}

\usepackage{amsmath,amsfonts,amsthm,amssymb,here}
\usepackage{hyperref}
\usepackage{amsmath}
\usepackage{graphicx}
\begin{document}
\bibliographystyle{plain}
\newtheorem{theo}{Theorem}[section]
\newtheorem{lemme}[theo]{Lemma}
\newtheorem{cor}[theo]{Corollary}
\newtheorem{defi}[theo]{Definition}
\newtheorem{prop}{Proposition}
\newtheorem{problem}[theo]{Problem}
\newtheorem{remarque}[theo]{Remark}
\newtheorem{claim}[theo]{Claim}
\newcommand{\beq}{\begin{eqnarray}}
\newcommand{\enq}{\end{eqnarray}}
\newcommand{\be}{\begin{eqnarray*}}
\newcommand{\en}{\end{eqnarray*}}
\newcommand{\ben}{\begin{eqnarray*}}
\newcommand{\enn}{\end{eqnarray*}}
\newcommand{\Td}{\mathbb T^d}
\newcommand{\Rd}{\mathbb R^n}
\newcommand{\R}{\mathbb R}
\newcommand{\N}{\mathbb N}
\newcommand{\Sn}{\mathbb S}
\newcommand{\Zd}{\mathbb Z^d}
\newcommand{\Linf}{L^{\infty}}
\newcommand{\dt}{\partial_t}
\newcommand{\Dt}{\frac{d}{dt}}
\newcommand{\Dtt}{\frac{d^2}{dt^2}}
\newcommand{\demi}{\frac{1}{2}}
\newcommand{\vf}{\varphi}
\newcommand{\epu}{_{\epsilon}}
\newcommand{\ep}{^{\epsilon}}
\newcommand{\bfi}{{\mathbf \Phi}}
\newcommand{\bpsi}{{\mathbf \Psi}}
\newcommand{\bx}{{\mathbf x}}
\newcommand{\bX}{{\mathbf X}}
\newcommand{\ds}{\displaystyle}
\newcommand {\g}{\`}
\newcommand{\E}{\mathbb E}
\newcommand{\Q}{\mathbb Q}
\newcommand{\PP}{\mathbb P}
\newcommand{\1}{\mathbb I}
\let\cal=\mathcal

\title{Forecasting trends with asset prices}
\maketitle
\begin{center}
Ahmed Bel Hadj Ayed\footnotemark[1]$^{,}$\footnotemark[2], Gr\'egoire Loeper \footnotemark[2], Fr\'ed\'eric Abergel \footnotemark[1] 
\end{center}
\begin{abstract}
In this paper, we consider a stochastic asset price model where the trend is an unobservable Ornstein Uhlenbeck process.
We first review some classical results from Kalman filtering. Expectedly, the choice of the parameters is crucial to put it into practice. For this purpose, we obtain the likelihood in closed form, and provide two on-line computations of this function. Then, we investigate the asymptotic behaviour of statistical estimators. Finally, we quantify the effect of a bad calibration with the continuous time mis-specified Kalman filter. Numerical examples illustrate the difficulty of trend forecasting in financial time series. 
\end{abstract}
\footnotetext[1]{Chaire of quantitative finance, laboratory MAS, CentraleSup\'elec}
\footnotetext[2]{BNP Paribas Global
Markets}

\section*{Motivations}
Asset prices are well described by random walks. The underlying economic foundation, introduced in 1900 by Louis Bachelier (see \cite{Bachelier}), is that, in an efficient market, price changes reflect new information and hence approximate a random variation. The first analytical and realistic approach was proposed in 1959 by Osborne (see \cite{Osborne}), who modelled returns as an ordinary random walk. In this case, future returns are not predictable. Nevertheless, people have divergent views on the subject and trend following strategies are the principal sources of returns for Commodity Trading Advisors (see \cite{Bouch}). 
Furthermore, most quantitative strategies are based on the assumption of trends extracted from the asset prices (see \cite{Markov}, \cite{Merton}). This component contains information about global changes and is useful for prediction. 
This estimation is a very difficult problem because of a high measurement noise. Consider for example the simple model with a constant trend $\frac{dS_{t}}{S_{t}}=\mu dt+\sigma_{S} dW_{t}^{S}$: the best trend estimate at time $t$ is $\widehat{\mu}_t=\frac{1}{t}\int_{0}^{t}\frac{dS_{s}}{S_{s}}$. The Student's t-test will reject the hypothesis $\mu=0$ at time $T$ if $|\widehat{\mu}_T| > \frac{q_\alpha \sigma_{S}}{\sqrt{T}} $, with $q_\alpha>1$ ($q_\alpha=1.96$ for a significance level equal to $\alpha=5\%$).
For instance, with $\sigma_{S}=30\%$, the estimate $\widehat{\mu}_T=1\%$ is  statistically relevant if $T > q_\alpha^{2}\ 900$ years. 
Generally, unobserved stochastic trends are considered and filtering methods are used (see \cite{LaknerStrat}, \cite{PhamStrat},\cite{Lions} or \cite{PortfolioBrendle}). Most of these filters introduce a parametric trend model. Therefore, these methods are confronted to the choice of the parameters.
Based on the Bayesian approach (the Kalman filter, the Extended Kalman filter or the Particle filter for example) or on the maximum likelihood estimation (see \cite{MLEstimMarkov},\cite{CAPPE}, \cite{ParamsEstimation}, \cite{ParamsEstimationINRIA} or \cite{ParamsEstimParticle}), several studies on the inference of hidden processes have been done in the past.
Many researchers applied these methods on financial time series, but they often focus on stochastic volatility processes (see \cite{vol1},\cite{vol2},\cite{vol3} or \cite{vol4}). 

The purpose of this work is to assess the feasibility of forecasting trends modeled by an unobserved mean-reverting diffusion.

The paper is organized as follows: the first section presents the model and recalls some results from Kalman filtering. 

The second section is devoted to the inference of the parameters with discrete time observations. Since all statistical estimators are based on the likelihood, two recursive computations of this function are presented.
Based on the Valentine Genon-Catalot's results (see \cite{LectureA} for details), the performance of statistical estimators is evaluated by giving their asymptotic behaviours, and by providing, in closed form, the Cramer Rao bound. 

In the third section, we introduce the continuous time mis-specified Kalman filter (like in \cite{MisFilter}, where the authors consider the case of an asset whose trend changes at
unknown random times), which takes into account a bad calibration of the parameters. First, the impact of parameters mis-specification on trend filtering is quantified with the law of the residuals between the filter (mis-specified or not) and the hidden process. Then, we derive the probability to have a positive trend, knowing a positive estimate. Due to the non-zero correlation between the trend and the filter, this probability is always superior to $0.5$.

Finally, numerical examples illustrate the difficulty of this calibration and show that the impact of parameters mis-specification on trend filtering is not negligible.

\section{Framework}
This section begins by introducing the model, which corresponds to an unobserved mean-reverting diffusion. After discretizing the model, we recall the Kalman filtering method.
\subsection{Model}
\subsubsection{Continuous time model}
Consider a financial market living on a stochastic basis $( \Omega , \mathcal{F} ,\mathbf{F}, \mathbb{P} )$, where  $\mathbf{F}=\left\lbrace  \mathcal{F}_{t} \right\rbrace$ is a filtration on $\Omega$, and $\mathbb{P}$ is the objective probability measure. Suppose that the dynamics of the risky asset $S$ is given by:
\begin{eqnarray}\label{Model}
\frac{dS_{t}}{S_{t}}&=&\mu_{t}dt+\sigma_{S} dW_{t}^{S},\\
d\mu_{t}&=&-\lambda_{\mu}\mu_{t}dt+\sigma_{\mu}dW_{t}^{\mu},
\end{eqnarray}
with $W^{S}$ and $W^{\mu}$ two uncorrelated Wiener processes and $\mu_{0}=0$. We also assume that $\left( \lambda_{\mu},\sigma_{\mu},\sigma_{S}\right) \in  \mathbb{R}_+^{*}\times\mathbb{R}_+^{*}\times\mathbb{R}_+^{*}$ and that the trend $\mu$ and the Brownian motion $W^{S}$ are $\mathbf{F}$-adapted. 
\begin{remarque}
Let $\mathbf{F}^{S}=\left\lbrace  \mathcal{F}_{t}^S \right\rbrace$ be the augmented filtration generated by the price process $S$. Only $\mathbf{F}^{S}$-adapted processes are observable, which implies that agents in this market do not observe the trend $\mu$.
\end{remarque}
\subsubsection{Discrete time model}
Let $\delta$ be a discrete time step. To simplify the notation, $k$ is used for $t_{k}=k\delta$. The discrete time model is:
\begin{eqnarray}\label{Equation Kalman}
y_{k+1}=\frac{S_{k+1}-S_{k}}{S_{k}\delta}&=& \mu_{k+1} + u_{k+1},\\
\mu_{k+1}&=&e^{-\lambda_{\mu} \delta} \mu_{k}+v_{k},
\end{eqnarray} 
where $u_{k}\sim \mathcal{N}\left( 0,\frac{\sigma_{S}^{2}}{\delta}\right)$ and $v_{k}\sim \mathcal{N}\left( 0,\frac{\sigma_{\mu}^{2}}{2\lambda_{\mu}}\left( 1- e^{-2 \lambda_{\mu} \delta} \right) \right)$.
The system \hyperref[Equation Kalman]{(3)-(4)} corresponds to an AR$(1)$ plus noise model.
\subsection{Optimal trend estimator} 
\subsubsection{Discrete Kalman filter}
In this sub section, the parameters $\theta=\left(\lambda_{\mu} ,\sigma_{\mu}\right)$ and $\sigma_{S}$ are supposed to be known. The discrete time system \hyperref[Equation Kalman]{(3)-(4)} corresponds to a Linear Gaussian Space State model where the observation is $y$ and the state $\mu$ (see \cite{ARMABook} for details). In this case, the Kalman filter gives the optimal estimator, which corresponds to the conditional expectation $\mathbb{E} \left[\mu_k|y_{1},...,y_{k}\right]$.
In the following, to simplify the notation, $\widehat{X}_{k/l}$ represents $\mathbb{E}\left[X_k|y_{1},...,y_{l}\right]$. Appendix \hyperref[sec::KalmanFilter]{A} presents a detailed introduction to the discrete Kalman filter. This filter contains two distinct phases:
\begin{enumerate}
\item An $a\ priori$ estimate given $\hat{\mu}_{k+1/k}$ and $\Gamma_{k+1/k}=\mathbb{E} \left[ (\mu_{k+1}- \right. $ $ \left.\hat{\mu}_{k+1/k})(\mu_{k+1}- \hat{\mu}_{k+1/k})^{T}\right]$. This estimate is done using the transition equation \hyperref[Equation Kalman]{$\left(4\right) $}.
\item An $a\ posteriori$ estimate. When the new observation is available, a correction of the first estimate is done to obtain $\hat{\mu}_{k+1/k+1}$ and $\Gamma_{k+1/k+1}=
\mathbb{E} \left[ (\mu_{k+1}-\hat{\mu}_{k+1/k+1})(\mu_{k+1} -\hat{\mu}_{k+1/k+1})^{T}\right]$. The criterion for this correction is the least squares method. 
\end{enumerate}
Thus, $\hat{\mu}_{k/k}$ is the minimum variance linear unbiased estimate of the trend $\mu_k$. To this end, it is the optimal recursive state estimate for a Linear Gaussian model.
Formally, the iterative method is given by:
\begin{eqnarray}
\widehat{\mu}_{k+1/k+1}&=&e^{-\lambda_{\mu} \delta} \widehat{\mu}_{k/k}+K_{k+1} \left( y_{k+1} -  e^{-\lambda_{\mu} \delta} \widehat{\mu}_{k/k} \right),\label{KalmanRecursion1}\\
\Gamma_{k+1/k+1}&=&\left( 1- K_{k+1} \right) \Gamma_{k+1/k},\label{KalmanRecursion2}
\end{eqnarray}
with 
\begin{eqnarray*}
K_{k+1}&=&\frac{\Gamma_{k+1/k}}{\Gamma_{k+1/k}+\frac{\sigma_{S}^{2}}{\delta}},\\
\Gamma_{k+1/k}&=&e^{-2 \lambda_{\mu} \delta}\Gamma_{k/k}+\frac{\sigma_{\mu}^{2}}{2\lambda_{\mu}}\left( 1- e^{-2 \lambda_{\mu} \delta} \right).
\end{eqnarray*}
\subsubsection{Stationary limit and continuous time representation}
Looking for $\Gamma_{k+1/k+1}=\Gamma_{k/k}$ (which corresponds to the steady-state regime), we find:
\begin{eqnarray*}
\Gamma_{\infty}&=&\frac{g\left(\sigma_{S}, \lambda_{\mu},\sigma_{\mu}\right)-f\left(\sigma_{S}, \lambda_{\mu},\sigma_{\mu}\right) }{2e^{-2 \lambda_{\mu} \delta}},\\
\text{where } f\left(\sigma_{S}, \lambda_{\mu},\sigma_{\mu}\right) &=& \left(\frac{\sigma_{S}^{2}}{\delta}+ \frac{\sigma_{\mu}^{2}}{2\lambda_{\mu}}\right)\left( 1- e^{-2 \lambda_{\mu} \delta} \right),\\
\text{and } g\left(\sigma_{S}, \lambda_{\mu},\sigma_{\mu}\right)&=&\sqrt{f\left(\sigma_{S}, \lambda_{\mu},\sigma_{\mu}\right) ^{2}+\frac{2 \sigma_{S}^{2} \sigma_{\mu}^{2}}{\lambda_{\mu}\delta}\left( e^{-2 \lambda_{\mu} \delta}- e^{-4 \lambda_{\mu} \delta} \right)}.
\end{eqnarray*}
Using the stationary covariance error $\Gamma_{\infty}$, a stationary gain $K_{\infty}$ is defined  and the estimate can be rewritten as a corrected exponential average: 
\begin{eqnarray}
\widehat{\mu}_{n+1}=K_{\infty} \sum_{i=0}^{\infty} e^{-\lambda_{\mu} \delta i} \left( 1 - K_{\infty} \right)^{i}   y_{n+1-i},
\end{eqnarray}
The steady-state Kalman filter has also a continuous time limit depending on the asset returns:
\begin{prop}\label{ContinuousEstimateProposition}
\begin{eqnarray}\label{ContinuousEstimate}
d \widehat{\mu}_{t}=-\lambda_{\mu} \beta\left( \lambda_{\mu},\sigma_{\mu},\sigma_{S}\right) \widehat{\mu}_{t}dt+\lambda_{\mu}\left( \beta\left( \lambda_{\mu},\sigma_{\mu},\sigma_{S}\right) -1 \right) \frac{dS_{t}}{S_{t}}, 
\end{eqnarray}
where
\begin{eqnarray}\label{BetaFormula}
\beta\left( \lambda_{\mu},\sigma_{\mu},\sigma_{S}\right) = \left( 1+ \frac{\sigma_{\mu}^{2}}{ \lambda_{\mu}^{2} \sigma_{S}^{2}}  \right)^{\frac{1}{2}}.
\end{eqnarray}
\end{prop}
\begin{proof}
Based on \cite{LaknerStrat}, the Kalman filter is given by:
\begin{eqnarray*}
E\left[  \mu_t | \mathcal{F}^S_t\right]&=&\phi\left( t\right) \left(\widehat{\mu}_{0}+\frac{1}{\sigma_{S}^{2}}\int_{0}^{t}\frac{P\left( u\right) }{\phi\left( u\right)} \frac{dS_{u}}{S_{u}} \right),\\
\phi\left( t\right)&=&e^{-\lambda_{\mu} t-\frac{1}{\sigma_{S}^{2}} \int_{0}^{t}P\left( u\right) du},
\end{eqnarray*}
where the estimation error variance $P$ is the solution of the following Riccati equation:
\begin{eqnarray*}
P'\left( t\right)=\frac{-1}{\sigma_{S}^{2}}P\left( t\right)^2-2\lambda_{\mu} P\left( t\right)+\sigma_{\mu}^{2}.
\end{eqnarray*}
In this steady-state regime, we have $P'\left( t\right)=0$. Then, the positive solution of this equation is given by:
\begin{eqnarray*}
P^\infty=\sigma_{S}^{2}\lambda_{\mu}\left(\beta\left( \lambda_{\mu},\sigma_{\mu},\sigma_{S}\right)-1 \right),
\end{eqnarray*}
and the steady-state Kalman filter follows:
\begin{eqnarray*}
\widehat{\mu}_{t}&=&\phi^\infty\left( t\right) \left(\widehat{\mu}_{0}+\frac{1}{\sigma_{S}^{2}}\int_{0}^{t}\frac{P^\infty }{\phi^\infty\left( u\right)} \frac{dS_{u}}{S_{u}} \right),\\
\phi^\infty\left( t\right)&=&e^{- \lambda_{\mu} \beta\left( \lambda_{\mu},\sigma_{\mu},\sigma_{S}\right) t}.
\end{eqnarray*}
Since:
\begin{eqnarray*}
\frac{d\phi^\infty\left( t\right)}{\phi^\infty\left( t\right)}=-\lambda_{\mu} \beta\left( \lambda_{\mu},\sigma_{\mu},\sigma_{S}\right)dt,
\end{eqnarray*}
the steady-state Kalman filter satisfies the stochastic differential equation (\ref{ContinuousEstimate}).
\end{proof}
This continuous representation can be used for risk return analysis of trend following strategies (see \cite{Lyxor} for details).
The Kalman filter is the optimal estimator for linear systems with Gaussian uncertainty. In practice, the parameters $\theta=\left(\lambda_{\mu},\sigma_{\mu}\right)$ are unknown and must be estimated.

\section{Inference of the trend parameters}
In this section, the problem of the parameters inference is treated. Based on discrete time observations, two classes of methods can be considered. The first one is based on backtesting. Each set of parameters defines one trend estimator and can be applied to several trading strategies. Backtests can be used in order to choose the parameters but even if several in- and out-of-sample periods are used, it will not ensure a good fit of the model. An alternative and more rigorous way exists: the use of statistical estimators. For example, Maximum Likelihood and Bayesian estimators have good properties (like consistency). To this end, this second approach is considered. For the model \hyperref[Equation Kalman]{(3)-(4)}, Peter Lakner (see \cite{Lakner}) and Ofer Zeitouni (see \cite{Zeitouni}) develop methods based on the algorithm EM in order to attain the maximum likelihood estimator but Fabien Campillo and François Le Gland suggest that direct maximization should be preferred in some cases (see \cite{MleVsEm} for details). Since the Bayesian estimators are also based on the likelihood, we present two on-line computations of this function. Using the Valentine Genon-Catalot's results (see \cite{LectureA} for details), we close this section by analysing the asymptotic behaviours of statistical estimators and by providing the Cramer Rao bound in closed form.

\subsection{Likelihood Computation}
The likelihood can be computed using two methods. The first one is based on a direct calculus while the second method uses the Kalman filter.
\subsubsection{Direct computation of the likelihood}
A first approach is to directly compute the likelihood. The vectorial representation of the discrete time model \hyperref[Equation Kalman]{(3)-(4)} is:
\begin{eqnarray*}
\left(\begin{array}{cc}y_1\\\vdots\\y_N\\\end{array}\right)=\left(\begin{array}{cc}\mu_1\\\vdots\\
\mu_N\\\end{array}\right)+\left(\begin{array}{cc}
u_1\\\vdots\\u_N\\\end{array}\right),
\end{eqnarray*}
where $\left( \mu_1,\cdots,\mu_N\right)^T$ and $\left( u_1,\cdots,u_N\right)^T$, knowing $\theta=\left(\sigma_{\mu},\lambda_{\mu} \right) $, are two independent Gaussian processes. Therefore the vector $\left( y_1,\cdots,y_N\right)^T$, knowing $\theta$, is also a Gaussian process. The likelihood is then characterized by the mean $M_{y_{1:N}|\theta}$ and the covariance $\Sigma_{y_{1:N}|\theta}$:
\begin{eqnarray}\label{Cov}
M_{y_{1:N}|\theta}&=& 0\text{ } \left( \mu_{0}=0 \text{ is supposed}\right) ,\\
\Sigma_{y_{1:N}|\theta}&=&\Sigma_{\mu_{1:N}|\theta}+\Sigma_{u_{1:N}|\theta},
\end{eqnarray} 
where $\Sigma_{u_{1:N}|\theta}=\frac{\sigma_{S}^{2}}{\delta}I_N$ and $\Sigma_{\mu_{1:N}|\theta}=\left( \mathbb{C}ov\left( \mu_{t},\mu_{s}\right)  \right)_{1\leq t,s \leq N}$. Since the drift $\mu$ is an Ornstein Uhlenbeck process, then:
\begin{eqnarray}\label{OuCov}
\mathbb{C}ov\left( \mu_{t},\mu_{s}\right)&=&\frac{\sigma_{\mu}^{2}}{2\lambda_{\mu}} e^{-\lambda_{\mu}\left(s+t\right) }\left(e^{2\lambda_{\mu} s \wedge t}-1 \right).
\end{eqnarray} 
Finally, the likelihood is given by:
\begin{eqnarray}\label{Likelihood}
f\left(y_{1},...y_{N} | \theta \right)=&\frac{1}{\left( 2\pi\right)^{N/2} \sqrt{det\Sigma_{y_{1:N}|\theta}} }e^{\left(  \frac{-1}{2} \left(y_{1},...,y_{N} \right) \Sigma_{y_{1:N}|\theta}^{-1} \left(y_{1},...,y_{N} \right)^{T} \right)}. 
\end{eqnarray} 
\begin{remarque}
When the dimension $N$ is large, it is extremely difficult to directly invert the covariance matrix $\Sigma_{y_{1:N}|\theta}$ and to compute the determinant of this matrix. An iterative approach can be used instead and the details of which are given in Appendix \hyperref[sec::IterativeLikelihood]{B}.
\end{remarque}
\subsubsection{Computation of the likelihood using the Kalman filter}
The likelihood can also be evaluated via the prediction error decomposition (see \cite{KalmanML} for details):
\begin{eqnarray*}
f\left(y_{1},...y_{N}| \theta \right)&=&f\left(y_{N}| y_{1},...y_{N-1},\theta \right)f\left( y_{1},...y_{N-1}|\theta \right).\\
&=&\prod_{n=1}^{N}f\left(y_{n}| y_{1},...y_{n-1},\theta \right),
\end{eqnarray*}
and the conditional laws are given by the following proposition:
\begin{prop}
The process $\left(y_{n}| y_{1},...y_{n-1},\theta\right)$ is a Gaussian:
\begin{eqnarray*}
\left(y_{n}| y_{1},...y_{n-1},\theta\right)& \sim & \mathcal{N}\left( M_{y_{n|n-1}},\mathbb{V}ar_{y_{n|n-1}}\right),
\end{eqnarray*}
and
\begin{eqnarray*}
M_{y_{n|n-1}}&=&e^{-\lambda \delta} \hat{\mu}_{n-1/n-1},\\
\mathbb{V}ar_{y_{n|n-1}}&=& e^{-2\lambda \delta}\Gamma_{n-1/n-1}+\frac{\sigma_{\mu}^{2}}{2\lambda}\left(1- e^{-2\lambda \delta}\right) +\frac{\sigma_{S}^{2}}{\delta}.
\end{eqnarray*}
The $a\ posteriori$ estimate of the trend $\hat{\mu}_{n-1/n-1}$ and the covariance error $\Gamma_{n-1/n-1}$ are given by Kalman filtering (see Equations (\ref{KalmanRecursion1}) and (\ref{KalmanRecursion2})). 
\end{prop}
\begin{proof}
Since the process $y_{n}$ is a Gaussian, the process $\left(y_{n}| y_{1},...y_{n-1},\theta\right)$ is also Gaussian. Moreover, using Equations \hyperref[Equation Kalman]{$(3)-(4)$}, we have:
\begin{eqnarray*}
M_{y_{n|n-1}}&=&\hat{\mu}_{n/n-1}+0,\\
\hat{\mu}_{n/n-1}&=& e^{-\lambda \delta} \hat{\mu}_{n-1/n-1}+0,
\end{eqnarray*}
and
\begin{eqnarray*}
\mathbb{V}ar_{y_{n|n-1}}&=& \Gamma_{n/n-1} +\frac{\sigma_{S}^{2}}{\delta},\\
 \Gamma_{n/n-1} &=& e^{-2\lambda \delta}\Gamma_{n-1/n-1}+\frac{\sigma_{\mu}^{2}}{2\lambda}\left(1- e^{-2\lambda \delta}\right).
\end{eqnarray*}
\end{proof}
\begin{remarque}
In practice, the volatility is not constant. However, if the volatility $\sigma_{S}$ is $\mathbf{F}^{S}$-adapted, the two methods can be adapted and implemented. This assumption is satisfied if the volatility is a continuous time process. 
\end{remarque}
\subsection{Performance of statistical estimators}
In this sub-section, the asymptotic behaviour of the classical estimators is investigated. 
\subsubsection{Asymptotic behaviour of statistical estimator}
The discrete time model \hyperref[Equation Kalman]{(3)-(4)} can be reformulated using the following proposition (see \cite{LectureA} for details):
\begin{prop}\label{ArmaProperty}
Consider the model \hyperref[Equation Kalman]{(3)-(4)} with $\left( \lambda_{\mu},\sigma_{\mu},\sigma_{S}\right) \in  \mathbb{R}_+^{*}\times\mathbb{R}_+^{*}\times\mathbb{R}_+^{*}$. In this case, the process  $\left( y_i\right)$ is ARMA$\left( 1,1\right)$.
\end{prop}
The asymptotic behaviour of the classical estimators follows. 
Indeed, the identifiability property and the asymptotic normality of the maximum likelihood estimator are well known in stationary ARMA Gaussian processes (see \cite{ARMABook2}, section 10.8). Moreover, the asymptotic behaviour of the Bayesian estimators are also guaranteed by the ARMA$(1,1)$ property of the process $\left( y_i\right) $. If the prior density function  is continuous and positive in an open neighbourhood of the real parameters, the Bayesian estimators are asymptotically normal (see \cite{Bernstein} in which a generalized Bernstein Von Mises theorem for stationary "short memory" processes is given, or \cite{ARMABayes} for a discussion on the Bayesian analysis of ARMA processes).
\subsubsection{Cramer Rao bound}
This bound is the lowest variance of the unbiased estimators. 
The following corollary of the Cramer Rao bound Theorem gives a formal description of the CRB.
\begin{cor}
Consider the model \hyperref[Equation Kalman]{(3)-(4)} and $N$ observations $\left( y_1,\cdots\right.$ $ \left. ,y_N\right)^T$. Suppose that $\left( \lambda_{\mu},\sigma_{\mu},\sigma_{S}\right) \in  \mathbb{R}_+^{*}\times\mathbb{R}_+^{*}\times\mathbb{R}_+^{*}$. If $\hat{\theta}_N$ is an unbiased estimator of $\theta=\left(\lambda_{\mu},\sigma_{\mu}\right)$, we have:
\begin{eqnarray*}
\mathbb{C}ov_\theta \left( \hat{\theta}_N\right)\geqslant CRB\left(\theta \right).
\end{eqnarray*} 
This bound is given by $CRB\left(\theta \right)= I^{-1}_N\left( \theta\right)$, where $I_N\left( \theta\right)$ is the Fisher Information matrix:
\begin{eqnarray*}
\left( I_N\left( \theta\right)\right)_{i,j} =-\mathbb{E} \left[ \frac{\partial^2 \log f\left(y_{1},...y_{N} | \theta \right) }{\partial \theta_i \partial \theta_j}  \right],
\end{eqnarray*}
and $I_N\left( \theta\right)=N I_1\left( \theta\right)$. Moreover, the maximum likelihood estimator $\widehat{\theta}^{ML}_N$ attains this bound:
\begin{eqnarray*}
\sqrt N \left( \widehat{\theta}^{ML}_N-\theta\right) \rightarrow \mathcal{N}\left( 0,I^{-1}_1\left( \theta\right)  \right).
\end{eqnarray*}
\end{cor}
This result is a consequence of Proposition \ref{ArmaProperty} (see \cite{ARMABook2}, section 10.8). We can also provide an analytic representation of the Fisher information matrix:
\begin{theo}\label{TheoCRB}
For the model \hyperref[Equation Kalman]{(3)-(4)}, if $\left( \lambda_{\mu},\sigma_{\mu},\sigma_{S}\right) \in  \mathbb{R}_+^{*}\times\mathbb{R}_+^{*}\times\mathbb{R}_+^{*}$, we have:
\begin{eqnarray*}
I_1\left( \theta\right)=\left( \dfrac{1}{4\Pi} \int_{-\Pi}^{\Pi} f_{\theta}^{-2} \left(\omega \right)\dfrac{\partial f_{\theta} }{\partial \theta_i}\left(\omega \right)\dfrac{\partial f_{\theta} }{\partial \theta_j}\left(\omega \right)d \omega \right)_{1 \leq i,j \leq 2}, 
\end{eqnarray*}
where $f_{\theta}$ is the spectral density of the process $\left( y_i\right)$:
\begin{eqnarray*}
f_{\theta}\left(\omega \right) =\frac{\frac{\sigma_{\mu}^{2}}{2\lambda_{\mu}}\left( 1- e^{-2 \lambda_{\mu} \delta}\right) +\frac{\sigma_{S}^{2}}{\delta}\left( 1+e^{-2\lambda_\mu\delta}\right) -\frac{2e^{-\lambda_\mu\delta}\sigma_{S}^{2}}{\delta}\cos\left(\omega \right)  }{1+e^{-2\lambda_\mu\delta}-2e^{-\lambda_\mu\delta}\cos\left( \omega\right) }.
\end{eqnarray*}
\end{theo} 
\begin{proof}
The Whittle's formula (see \cite{Whittle} for details) gives the integral representation of the Fisher information matrix. Since the process  $\left( y_i\right)$ is ARMA$\left( 1,1\right)$, the expression of its spectral density follows (see \cite{ARMABook2}, section 4.4).
\end{proof}
Finally, the Cramer Rao Bound of the trend parameters can be computed using Theorem \ref{TheoCRB}.

\section{Impact of parameters mis-specification}
In this section, we consider the continuous time Kalman filter with a bad calibration in the steady-state regime. First, we characterize the law of the residuals between the filter (mis-specified or not) and the hidden process. Finally, we study the impact of parameters mis-specification on the detection of a positive trend.
\subsection{Context}	
Suppose that the risky asset $S$ is given by the model \hyperref[Model]{(1)-(2)} with $\theta^*=\left(\sigma_{\mu}^*,\lambda_{\mu}^* \right) $, and suppose that an agent thinks that the parameters are equal to $\theta=\left(\sigma_{\mu},\lambda_{\mu} \right) $. Assuming the steady-state regime and using these estimates and Proposition \ref{ContinuousEstimateProposition}, the agent implements the continuous time mis-specified Kalman filter:
\begin{eqnarray}\label{FauxFiltreEDS}
d \widehat{\mu}_{t}=-\lambda_{\mu} \beta \widehat{\mu}_{t}dt+\lambda_{\mu}\left( \beta -1 \right) \frac{dS_{t}}{S_{t}}, 
\end{eqnarray}
where $\beta=\beta\left( \lambda_{\mu},\sigma_{\mu},\sigma_{S}\right)$ (see Equation (\ref{BetaFormula})) and $\widehat{\mu}_{0}=0$.
The following lemma gives the law of the mis-specified Kalman filter: 
\begin{lemme}\label{VarianceFauxFiltre}
Consider the model \hyperref[Model]{(1)-(2)} with $\theta^*=\left(\sigma_{\mu}^*,\lambda_{\mu}^* \right) $. In this case, the mis-specified continuous time filter of Equation (\ref{FauxFiltreEDS}) is given by:
\begin{eqnarray}\label{False Filter}
\widehat{\mu}_{t}=\lambda_{\mu}\left(\beta -1 \right)e^{-\lambda_{\mu} \beta t} \left(   \int_{0}^{t}e^{\lambda_{\mu} \beta s}\mu^{*}_{s}ds+\sigma_{S} \int_{0}^{t}e^{\lambda_{\mu} \beta s}dW_{s}^{S} \right). 
\end{eqnarray}
Moreover, $\widehat{\mu}$ is a centered Gaussian process and its variance is given by:
\begin{eqnarray*}
&&\mathbb{V}ar \left[ \widehat{\mu}_{t}\right] =\mathbb{E} \left[ \widehat{\mu}_{t}^{2}\right] =\frac{\lambda_{\mu}^2 \left(\beta -1 \right)^2 \left( \sigma_{\mu}^{*}\right)^2 }{\lambda_{\mu}^{*}\left( \lambda_{\mu}\beta-\lambda_{\mu}^{*}\right) }\left[ \frac{1-e^{-\left( \lambda_{\mu}\beta+\lambda_{\mu}^{*}\right)t }}{\lambda_{\mu}\beta+\lambda_{\mu}^{*}}\right. \\
&&\left. +\frac{2e^{-\left( \lambda_{\mu}\beta+\lambda_{\mu}^{*}\right)t }-e^{-2\lambda_{\mu}^{*}t}-e^{-2 \lambda_{\mu}\beta t }   }{\lambda_{\mu}\beta-\lambda_{\mu}^{*}}+\frac{e^{-2 \lambda_{\mu}\beta t}-1}{2 \lambda_{\mu}\beta} \right] \\ &&+
\frac{\lambda_{\mu}\left(\beta -1 \right)^2  \sigma_{S}^2 }{2\beta }\left(1-e^{-2 \lambda_{\mu}\beta t} \right). 
\end{eqnarray*}
\end{lemme}
\begin{proof}
Applying It\^o's lemma to the function $f\left(\widehat{\mu},t\right) =e^{\lambda_{\mu} \beta t}\widehat{\mu}_{t}$, and integrating from $0$ to $t$, Equation (\ref{False Filter}) follows. Therefore, $\widehat{\mu}$ is also a Gaussian process. Its mean is zero (because $\mu^{*}_{0}=0$). Since the processes $\mu^{*}$ and $W^{S}$ are supposed to be independent, the variance of $\widehat{\mu}$ is given by the sum of the variances of the terms in Equation (\ref{False Filter}). Moreover:
\begin{eqnarray*}
\mathbb{V}ar \left[ \int_{0}^{t}e^{\lambda_{\mu}\beta}dW_{s}^{S}\right]&=&\frac{e^{2\lambda_{\mu}\beta t}-1}{2 \lambda_{\mu}\beta},\\
\mathbb{V}ar \left[\int_{0}^{t}e^{\lambda_{\mu} \beta s}\mu^{*}_{s}ds\right] &=&\int_{0}^{t}\int_{0}^{t} e^{\lambda_{\mu} \beta \left( s_1+s_2\right) }\mathbb{C}ov\left( \mu^{*}_{s_1},\mu^{*}_{s_2}\right)ds_1 ds_2,
\end{eqnarray*}
and $\mathbb{C}ov\left( \mu^{*}_{s_1},\mu^{*}_{s_2}\right)$ is given by Equation (\ref{OuCov}). The variance of the process $\widehat{\mu}_{t}$ follows.
\end{proof}
\subsection{Filtering with parameters mis-specification}
The impact of parameters mis-specification on trend filtering can be measured using the difference between the filter and the hidden process. 
The following theorem gives the law of the residuals.
\begin{theo}\label{TheoremMiss}
Consider the model \hyperref[Model]{(1)-(2)} with $\theta^*=\left(\sigma_{\mu}^*,\lambda_{\mu}^* \right) $ and the trend estimate defined in Equation (\ref{False Filter}). In this case, the process $\widehat{\mu}-\mu^{*}$
is a centered Gaussian process and its variance has a stationary limit:
\begin{eqnarray}\label{EquationAsymptoticResiduals}
  \lim\limits_{t\rightarrow\infty}\mathbb{V}ar \left[ \widehat{\mu}_{t}-\mu^{*}_{t}\right] =
\frac{\sigma_S^2}{2\beta}\left( \lambda_{\mu}\left(\beta-1 \right)^2 +\lambda_{\mu}^*\left(\left( \beta^*\right) ^2-1  \right)\frac{\lambda_{\mu}^*\beta+\lambda_{\mu}}{\lambda_{\mu}\beta+\lambda_{\mu}^*} \right),
\end{eqnarray}
where $\beta=\beta\left( \lambda_{\mu},\sigma_{\mu},\sigma_{S}\right)$ and $\beta^*=\beta\left( \lambda_{\mu}^*,\sigma_{\mu}^*,\sigma_{S}\right)$ (see Equation (\ref{BetaFormula})). 

Moreover, if $\left(\sigma_{\mu},\lambda_{\mu} \right)=\left(\sigma_{\mu}^*,\lambda_{\mu}^* \right)$, Equation (\ref{EquationAsymptoticResiduals}) becomes:
\begin{eqnarray}\label{EquationAsymptoticResidualsWell}
   \lim\limits_{t\rightarrow\infty}\mathbb{V}ar \left[ \widehat{\mu}_{t}^* -\mu^{*}_{t}\right] =
\lambda_{\mu}^* \sigma_{S}^2\left( \beta^*-1\right).
\end{eqnarray}
\end{theo}
\begin{proof}
Using Equation (\ref{False Filter}), it follows that the process $\widehat{\mu}-\mu^{*}$
is a centered Gaussian process. 
The variance of this difference can be computed in closed form:
\begin{eqnarray*}
\mathbb{V}ar \left[ \widehat{\mu}_{t}-\mu^{*}_{t}\right]=\mathbb{V}ar \left[ \widehat{\mu}_{t}\right]+\mathbb{V}ar \left[ \mu^{*}_{t}\right]-2*\mathbb{C}\text{ov} \left[ \widehat{\mu}_{t},\mu^{*}_{t}\right],
\end{eqnarray*}
where $\mathbb{V}ar \left[ \mu^{*}_{t}\right]=\frac{\left( \sigma_{\mu}^{*}\right)^{2} }{2\lambda_{\mu}^{*}}\left(1-e^{-2\lambda_{\mu}^{*}t} \right) $, and $\mathbb{V}ar \left[ \widehat{\mu}_{t}\right]$ is given by Lemma \ref{VarianceFauxFiltre}. Since the processes $W^S$ and $\mu^{*}$ are supposed to be independent, we have:
\begin{eqnarray*}
\mathbb{C}\text{ov} \left[ \widehat{\mu}_{t},\mu^{*}_{t}\right]&=&\frac{\lambda_{\mu}\left(\beta-1 \right)\left( \sigma_{\mu}^{*}\right)^{2} }{2\lambda_{\mu}^{*}}\left(\frac{1-e^{-\left( \lambda_{\mu}\beta+\lambda_{\mu}^{*}\right)t }}{\lambda_{\mu}\beta+\lambda_{\mu}^{*}}\right. \\
&&\left. -\frac{e^{-2\lambda_{\mu}^{*}t}-e^{-\left( \lambda_{\mu}\beta+\lambda_{\mu}^{*}\right)t }}{ \lambda_{\mu}\beta-\lambda_{\mu}^{*} } \right) .
\end{eqnarray*}
The asymptotic variance is obtained by tending t to infinity:
\begin{eqnarray*}
\lim\limits_{t\rightarrow\infty}\mathbb{V}ar \left[ \widehat{\mu}_{t}-\mu^{*}_{t}\right]&=&
\frac{ \lambda_{\mu} \left(\beta -1 \right)}{2\beta }\left[ \left(\beta -1 \right)\sigma_S^2-\frac{\left( \sigma_{\mu}^{*}\right)^{2}\left( \beta +1\right) }{\lambda_{\mu}^{*}\left( \lambda_{\mu}\beta+\lambda_{\mu}^{*}\right) }
 \right] \\
 && + \frac{\left( \sigma_{\mu}^{*}\right)^{2} }{2\lambda_{\mu}^{*}},
\end{eqnarray*}
and Equation (\ref{EquationAsymptoticResiduals}) follows. Finally, Equation (\ref{EquationAsymptoticResidualsWell}) is obtained by tending $\theta$ to $\theta^*$.
\end{proof}
\begin{remarque}
Consider the well-specified case $\left(\sigma_{\mu},\lambda_{\mu} \right)=\left(\sigma_{\mu}^*,\lambda_{\mu}^* \right)$. Using Equation (\ref{EquationAsymptoticResidualsWell}), it follows that:
\begin{eqnarray}\label{RatioVariance}
 \lim\limits_{t\rightarrow\infty}\frac{\mathbb{V}ar \left[ \widehat{\mu}_{t}^* -\mu^{*}_{t}\right]}{\mathbb{V}ar \left[\mu^{*}_{t}\right]}=\frac{2}{1+\sqrt{1+ \frac{\left( \sigma_{\mu}^{*}\right) ^{2}}{ \left( \lambda_{\mu}^{*}\right) ^{2} \sigma_{S}^{2}}}}.
\end{eqnarray}
Then, the asymptotic relative variance of the well-specified residuals is an increasing function of $\lambda_{\mu}^{*}$ and a decreasing function of $\sigma_{\mu}^{*}$. 
\end{remarque}
\subsection{Detection of a positive trend}
In practice, the trend estimate (mis-specified or not) can be used for an investment decision. For example, a positive estimate leads to a long position. So, it is interesting to know the probability to have a positive trend, knowing a positive estimate. We derive this probability in closed form.
The following proposition gives the asymptotic conditional law of the trend  $\left( \mu^{*}_t|\hat{\mu}_t=x\right)$:
\begin{prop}\label{ConditionalLaw}
Consider the model \hyperref[Model]{(1)-(2)} with $\theta^*=\left(\sigma_{\mu}^*,\lambda_{\mu}^* \right) $ and the trend estimate defined in Equation (\ref{False Filter}). In this case:
\begin{eqnarray}\label{CondDistrib}
\left( \mu^{*}_t|\hat{\mu}_t=x\right) \underset{t \rightarrow \infty}{\overset{\mathcal{L}}{\rightarrow}} \mathcal{N}\left( \mathbb{M}_{\mu^{*}|\hat{\mu}}^{\infty},\mathbb{V}ar_{\mu^{*}|\hat{\mu}}^{\infty}\right) ,
\end{eqnarray}
with:
\begin{eqnarray}
\mathbb{M}_{\mu^{*}|\hat{\mu}}^{\infty}&=&\frac{\lambda_\mu^*\beta\left(\left(\beta^* \right)^2-1  \right) }{\left( \beta-1\right)\left(\lambda_\mu \beta+\lambda_\mu^*\left(\beta^* \right)^2 \right)  }x,\label{eq1}\\ 
\mathbb{V}ar_{\mu^{*}|\hat{\mu}}^{\infty}&=&\mathbb{V}ar_{\mu^*}^{\infty}\left( 1-\frac{\lambda_\mu^*\lambda_\mu\beta\left(\left(\beta^* \right)^2-1  \right)}{\left(\lambda_\mu^*+\lambda_\mu\beta \right) \left(\lambda_\mu \beta+\lambda_\mu^*\left(\beta^* \right)^2 \right)}\right) ,\label{eq2}
\end{eqnarray}
where $\mathbb{V}ar_{\mu^*}^{\infty}=\frac{(\sigma_\mu^*)^2}{2\lambda_\mu^*}$. 

Moreover, if $\left(\sigma_{\mu},\lambda_{\mu} \right)=\left(\sigma_{\mu}^*,\lambda_{\mu}^* \right)$, Equation (\ref{CondDistrib}) becomes:
\begin{eqnarray}\label{CondDistribWellSpecified}
\left( \mu^{*}_t|\hat{\mu}^*_t=x\right) \underset{t \rightarrow \infty}{\overset{\mathcal{L}}{\rightarrow}} \mathcal{N}\left( x,\frac{2\mathbb{V}ar_{\mu^*}^{\infty}}{\beta^*+1}\right) ,
\end{eqnarray}
where $\beta^*=\beta\left( \lambda_{\mu}^*,\sigma_{\mu}^*,\sigma_{S}\right)$ (see Equation (\ref{BetaFormula})).
\end{prop}
\begin{proof}
Since the estimate $\hat{\mu}$ and the trend $\mu^{*}$ are two centred and correlated Gaussian processes (see Lemma \ref{VarianceFauxFiltre} and the proof of Theorem \ref{TheoremMiss}), the conditional law $\left( \mu^{*}_t|\hat{\mu}_t=x\right)$ is Gaussian with a mean and a variance given by:
\begin{eqnarray*}
\mathbb{M}_{\mu^{*}_t|\hat{\mu}_t}&=&\frac{\mathbb{C}ov\left(\hat{\mu}_t,\mu^{*}_t \right) }{\mathbb{V}ar\left[ {\hat{\mu}_t}\right] }x,\\
\mathbb{V}ar_{\mu^{*}_t|\hat{\mu}_t}&=&\mathbb{V}ar\left[ {\mu^*_t}\right] -\frac{\mathbb{C}ov\left(\hat{\mu}_t,\mu^{*}_t \right)^2 }{\mathbb{V}ar\left[ {\hat{\mu}_t}\right] }.
\end{eqnarray*}
Using Lemma \ref{VarianceFauxFiltre} and the expression of $\mathbb{C}ov\left(\hat{\mu}_t,\mu^{*}_t \right)$ in the proof of Theorem \ref{TheoremMiss}:
\begin{eqnarray*}
\lim\limits_{t \rightarrow \infty}	\mathbb{M}_{\mu^{*}_t|\hat{\mu}_t}&=&\mathbb{M}_{\mu^{*}|\hat{\mu}}^{\infty},\\
\lim\limits_{t \rightarrow \infty}	\mathbb{V}ar_{\mu^{*}_t|\hat{\mu}_t}&=&\mathbb{V}ar_{\mu^{*}|\hat{\mu}}^{\infty},
\end{eqnarray*}
and Equation (\ref{CondDistrib}) follows. Moreover, Equation (\ref{CondDistribWellSpecified}) is obtained by tending $\theta$ to $\theta^*$.
\end{proof}
The following proposition is a consequence of the previous proposition. It gives the asymptotic probability to have a positive trend, knowing a positive estimate equal to $x$.
\begin{prop}\label{PropositionProba}
	Consider the model \hyperref[Model]{(1)-(2)} with $\theta^*=\left(\sigma_{\mu}^*,\lambda_{\mu}^* \right) $ and the trend estimate defined in Equation (\ref{False Filter}). In this case:
	\begin{eqnarray}\label{Proba1}
		\lim\limits_{t \rightarrow \infty}	\mathbb{P}\left(\mu^{*}_t>0|\hat{\mu}_t=x \right)	=	\mathbb{P}_{\infty}\left(\mu^{*}>0|\hat{\mu}=x \right),
	\end{eqnarray}
	where
	\begin{eqnarray}\label{Proba}
	\mathbb{P}_{\infty}\left(\mu^{*}>0|\hat{\mu}=x \right) =1-\Phi\left(\frac{-\mathbb{M}_{\mu^{*}|\hat{\mu}=x}^{\infty}}{\sqrt{\mathbb{V}ar_{\mu^{*}|\hat{\mu}=x}^{\infty}}} \right), 
	\end{eqnarray}
	where $\mathbb{M}_{\mu^{*}|\hat{\mu}=x}^{\infty}$ and $\mathbb{V}ar_{\mu^{*}|\hat{\mu}=x}^{\infty}$ are defined in Equations (\ref{eq1}) and (\ref{eq2}), and $\Phi$ is the cumulative distribution function of the standard normal law. 
	
	Moreover, if $x>0$ and $\left(\sigma_{\mu},\lambda_{\mu} \right)=\left(\sigma_{\mu}^*,\lambda_{\mu}^* \right)$, this asymptotic probability becomes an increasing function of $\sigma_\mu^*$ and a decreasing function of $\lambda_\mu^*$.
\end{prop}
\begin{proof}
Equations (\ref{Proba1}) and (\ref{Proba}) follow from Proposition \ref{ConditionalLaw}. Now, consider  the well-specified case $\left(\sigma_{\mu},\lambda_{\mu} \right)=\left(\sigma_{\mu}^*,\lambda_{\mu}^* \right)$ and $x>0$. Using Equation (\ref{CondDistribWellSpecified}), it follows that:
\begin{eqnarray*}
\mathbb{V}ar_{\mu^{*}|\hat{\mu}^*=x}^{\infty}=f\left( \sigma_{\mu}^*,\lambda_{\mu}^*,\sigma_S\right), 
\end{eqnarray*}
where
\begin{eqnarray*}
f\left( \sigma_{\mu}^*,\lambda_{\mu}^*,\sigma_S\right) =\frac{\left( \sigma_{\mu}^*\right)^2 }{\lambda_{\mu}^*\left(1+\sqrt{1+\frac{\left( \sigma_{\mu}^*\right)^2}{\sigma_S^2\left( \lambda_{\mu}^*\right)^2}} \right) }.
\end{eqnarray*}
Since
\begin{eqnarray*}
\frac{\partial f}{\partial \lambda_{\mu}^*}\left( \sigma_{\mu}^*,\lambda_{\mu}^*,\sigma_S\right)&=&\frac{-\left( \sigma_{\mu}^*\right)^2 }{\left( \lambda_{\mu}^*\right) ^2\left(1+\sqrt{1+\frac{\left( \sigma_{\mu}^*\right)^2}{\sigma_S^2\left( \lambda_{\mu}^*\right)^2}}+\frac{\left( \sigma_{\mu}^*\right)^2}{\sigma_S^2} \right) }\leq 0,\\
\frac{\partial f}{\partial \sigma_{\mu}^*}\left( \sigma_{\mu}^*,\lambda_{\mu}^*,\sigma_S\right)&=&\frac{\lambda_{\mu}^* \sigma_{\mu}^* \sigma_S^2 \sqrt{1+\frac{\left( \sigma_{\mu}^*\right)^2}{\sigma_S^2\left( \lambda_{\mu}^*\right)^2}} }{\left( \sigma_{\mu}^*\right)^2 +\sigma_S^2\left( \lambda_{\mu}^*\right) ^2 } \geq 0,
\end{eqnarray*}
the asymptotic well-specified probability to have a positive trend, knowing a positive estimate equal to $x$ is an increasing function of $\sigma_\mu^*$ and a decreasing function of $\lambda_\mu^*$. 
\end{proof}
\begin{remarque}\label{RemarkProba}
This probability is an increasing function of $x$. Indeed, it is easier to detect the sign of the real trend with a high estimate than with a low estimate. Moreover, this probability is always superior to 0.5. This is due to the non-zero correlation between the trend and the filter.
As shown in the previous sections, trend filtering is easier with a small spot volatility. Here, the probability to make a good detection is also a decreasing function of the $\sigma_S$. 
\end{remarque}
\section{Simulations}
In this section, numerical examples are computed in order to make the reader aware of the trend filtering problem. First, the feasibility of trend forecasting with statistical estimator is illustrated on different trend regimes. Then, the effects of a bad forecast on trend filtering and on the detection of a positive trend are also discussed.
\subsection{Feasibility of trend forecasting}
Suppose that only discrete time observations are available and that the discrete time step is equal to $\delta=1/252$. In this case, the agent uses the daily returns of the risky asset to calibrate the trend.
We also assume that the agent uses an unbiased estimator. Given $T$ years of observations, The Cramer Rao Bound is given by:
\begin{eqnarray*}
CRB_T\left( \theta\right)= \frac{I_1^{-1}\left( \theta\right)  }{T*252},
\end{eqnarray*}
where $I_1\left( \theta\right)$ is given by Theorem \ref{TheoCRB}. The smallest confidence region is obtained with this matrix. In practice, the real values of the parameters $\theta$ are unknown and asymptotic confidence regions are computed (replacing $\theta$ by the estimates $\hat{\theta}$ in the Fisher information matrix $I_1\left( \hat{\theta}\right)$). 
Since the goal of this subsection is to evaluated the feasibility of this problem, we suppose that we know the real values of the parameters. Then the real Cramer Rao Bound can be computed. Suppose that a target standard deviation $x_i$ is fixed for the parameter $\theta_i$. In this case, to reach the precision $x_i$, the length of the observations must be superior to: 
\begin{eqnarray*}
T^x_i = \frac{\left(I_1^{-1}\left( \theta\right) \right)_{ii} }{252*x_i^2}.
\end{eqnarray*}
We consider a fixed spot volatility $\sigma_S =30\%$, two target precisions for each parameter $\theta_i$ and we compute $T^x_i$ for several configurations. The figures \hyperref[Figu1]{1}, \hyperref[Figu2]{2}, \hyperref[Figu3]{3} and \hyperref[Figu4]{4} represent the results. It is well known that for a high measurement noise, which means a high spot volatility, the problem is harder because of a low signal-to-noise ratio. The higher the volatility, the longer the observations must be. Here, we observe that with a higher drift volatility $\sigma_{\mu}$ and a lower $\lambda_{\mu}$, the problem is easier. Indeed, the drift takes higher values and is more detectable. Moreover, the simulations show that the classical estimators are not adapted to a so weak signal-to-noise ratio. Even after a long period of observations,  the estimators exhibit high variances. Indeed the smallest period of observations is superior to 29 years. It corresponds to a target standard deviation equal to $0.5$ for a real parameter $\lambda_\mu=1$ and a trend standard deviation equal to $\sigma_\mu\left( 2\lambda_\mu\right)^{-1/2} \approx 63\%$. Therefore, for this configuration, after $30$ years of observations, the standard deviation is equal to $50\%$ of the real parameter value $\lambda_\mu$. After $742$ years, this standard deviation is equal to $10\%$.
Even with this kind of regime, the trend forecast with a good precision is impossible.
\begin{figure}[H]
\begin{center}
   \includegraphics[totalheight=7cm]{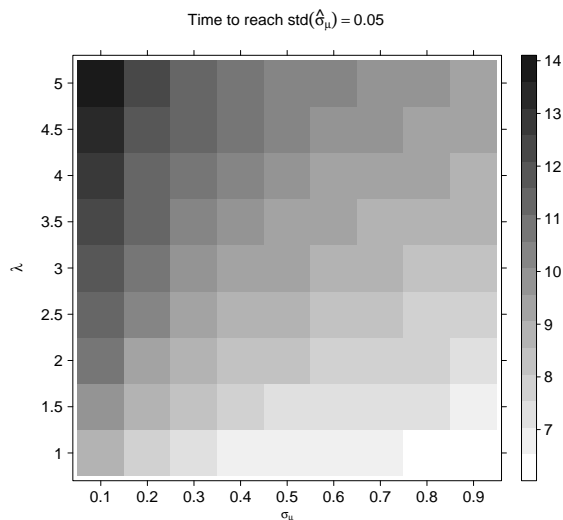}   
  \caption{Time to reach a target standard deviation on $\sigma_{\mu}$ equal to $0.05$ (ln(years)) }
\end{center}\label{Figu1}
\end{figure}
\begin{figure}[H]
\begin{center}
   \includegraphics[totalheight=7cm]{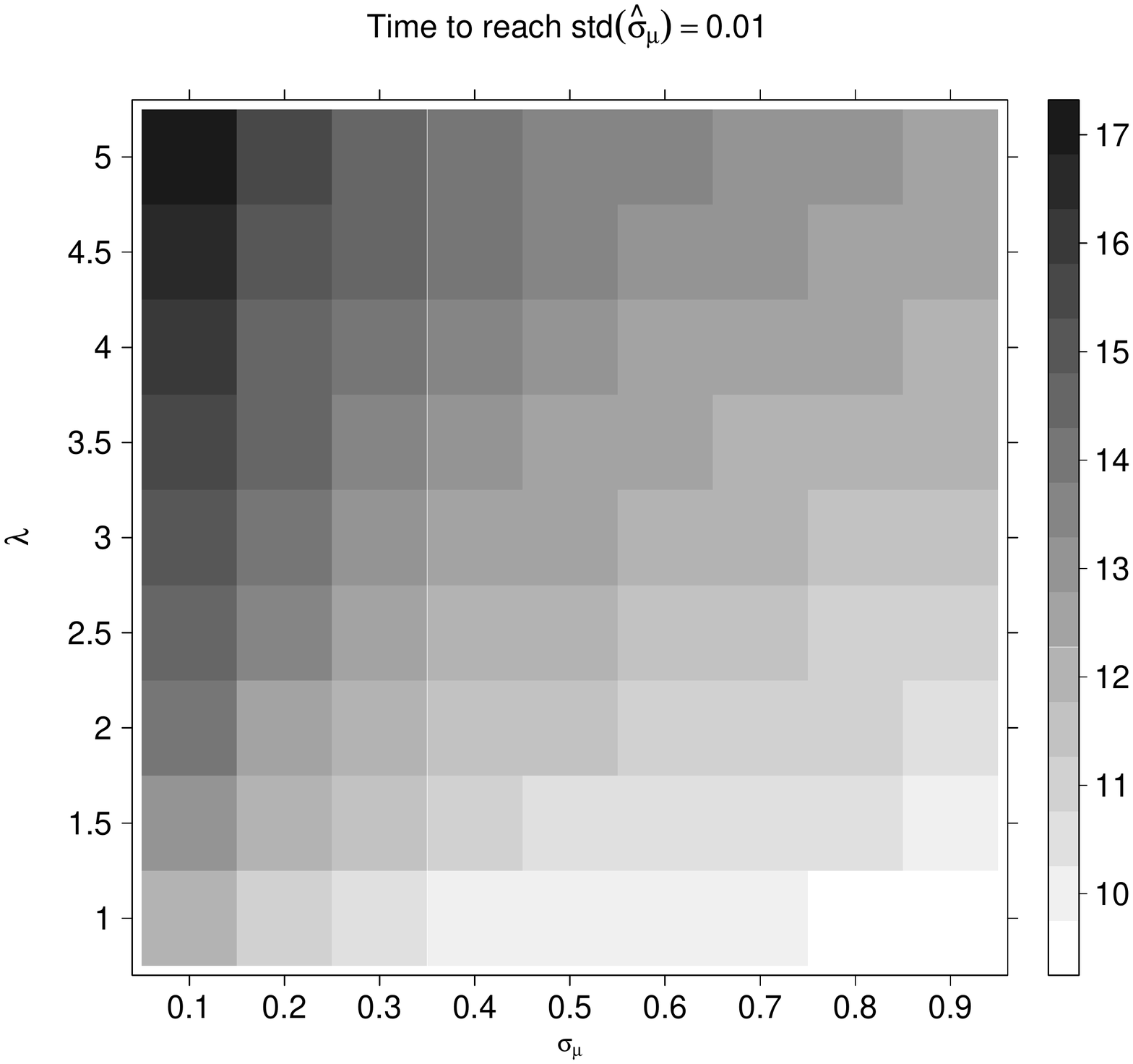}   
  \caption{Time to reach a target standard deviation on $\sigma_{\mu}$ equal to $0.01$ (ln(years)) }
\end{center}\label{Figu2}
\end{figure}
\begin{figure}[H]
\begin{center}
    \includegraphics[totalheight=7cm]{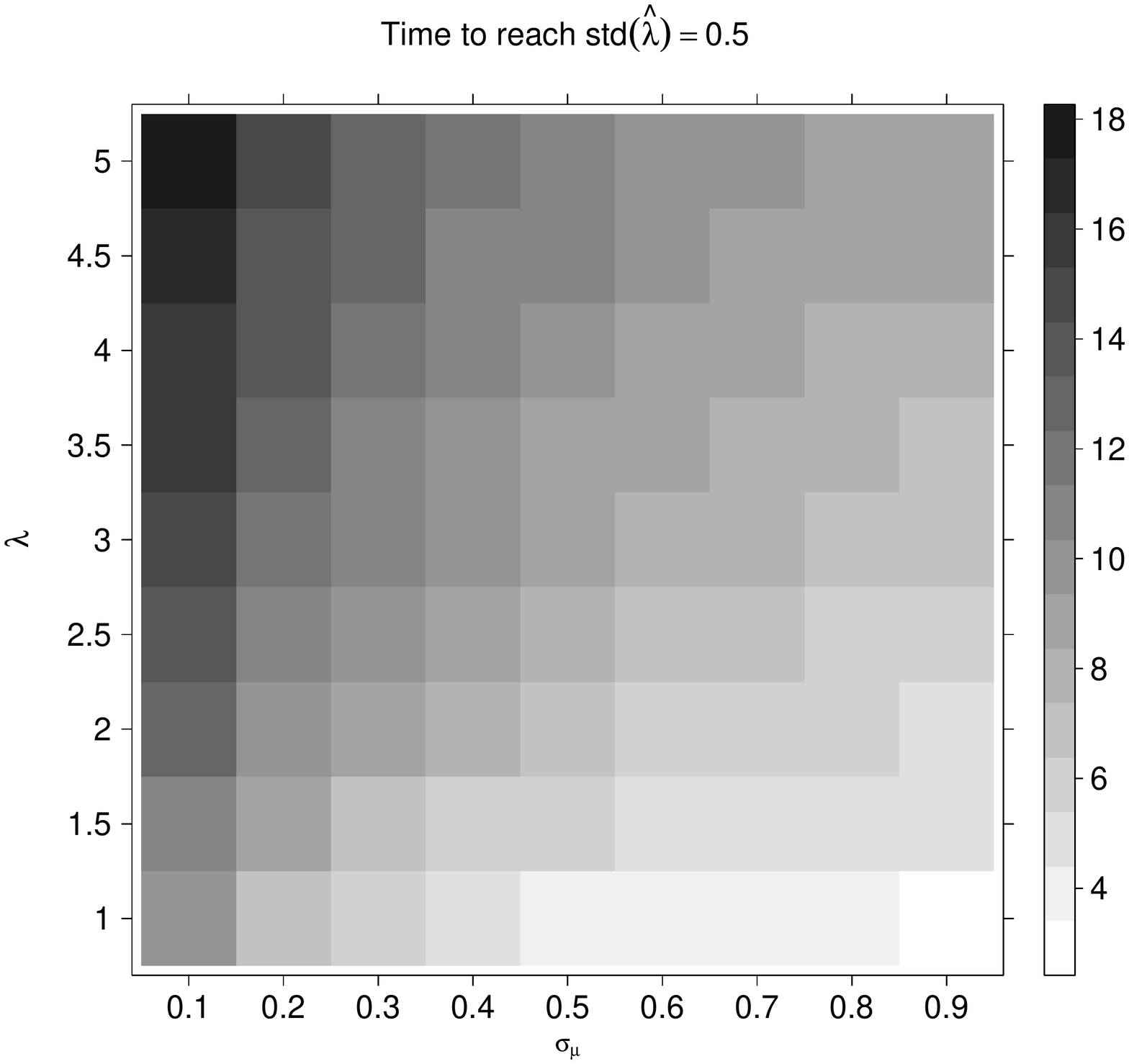}    
   \caption{Time to reach a target standard deviation on $\lambda_{\mu}$ equal to $0.5$ (ln(years)) }
\end{center}
\end{figure}\label{Figu3}
\begin{figure}[H]
\begin{center}
    \includegraphics[totalheight=7cm]{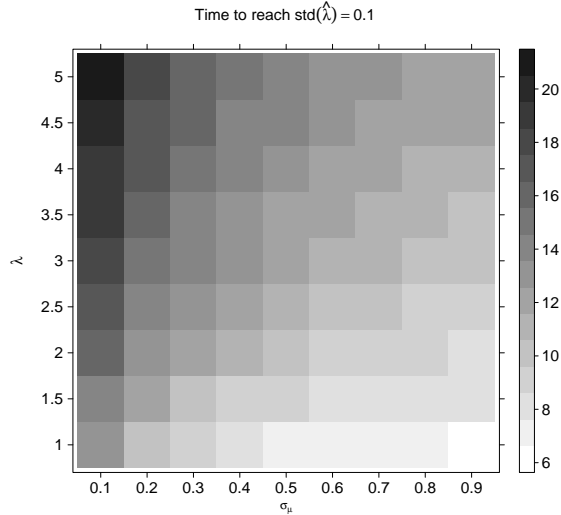}    
   \caption{Time to reach a target standard deviation on $\lambda_{\mu}$ equal to $0.1$ (ln(years)) }
\end{center}\label{Figu4}
\end{figure}

\subsection{Impact of parameters mis-specification on trend filtering}
This subsection illustrates the impact of parameters mis-specification on trend filtering. Using the results of Theorem \ref{TheoremMiss}, we represent, for different configurations, and for the well- and mis- specified case, the asymptotic standard deviation of the residuals between the trend and the filter. 
The figures \hyperref[Figu5]{5} and \hyperref[Figu6]{6} represent the asymptotic standard deviation of the trend and of the residuals in the well-specified case (the agent uses the real values of the parameters) for different configurations. As seen in Equation (\ref{EquationAsymptoticResidualsWell}), the asymptotic standard deviation of the well-specified residuals is an increasing function of the drift volatility $\sigma_{\mu}^*$ and a decreasing function of the parameter $\lambda_{\mu}^*$. 
For $\lambda_{\mu}^*=1$ and $\sigma_{\mu}^*=90\%$, the standard deviation of the residuals ($\simeq 44\%$) is inferior to the standard deviation of the trend ($\simeq 64\%$).
For a high $\lambda_{\mu}^*$ and a small drift volatility, the two quantities are approximately equal.  
This figure leads to the same conclusions than Equation (\ref{RatioVariance}). Indeed, like the calibration problem, the problem of trend filtering is easier with a small $\lambda_{\mu}^*$ and a high drift volatility $\sigma_{\mu}^*$.

 Now consider the worst configuration $\sigma_S=30\%$, $\lambda_{\mu}^{*}=5$ and $\sigma_{\mu}^{*}=10\%$. The figure \hyperref[Figu7]{7} represents the asymptotic standard deviation of the residuals for different estimates $\left( \lambda_{\mu},\sigma_{\mu}\right) $. This regime corresponds to a standard deviation of the trend equal to $\sigma^{*}_\mu\left( 2\lambda^{*}_\mu\right)^{-1/2} \approx 3.2\%$ and to a standard deviation of the residuals equal to $3.16\%$ in the well-specified case. If the agent implements the Kalman filter with  $\lambda_{\mu}=1$ and $\sigma_{\mu}=90\%$, the standard deviation of the residuals becomes superior to $25\%$. Finally, consider the best configuration $\sigma_S=30\%$, $\lambda_{\mu}^{*}=1$ and $\sigma_{\mu}^{*}=90\%$. The figure \hyperref[Figu8]{8} represents the asymptotic standard deviation of the residuals for different estimates $\left( \lambda_{\mu},\sigma_{\mu}\right) $. This regime corresponds to a trend standard deviation equal to $\sigma^{*}_\mu\left( 2\lambda^{*}_\mu\right)^{-1/2} \approx 63\%$ and to a standard deviation of the residuals equal to $44\%$ in the well-specified case. If the agent implements the Kalman filter with  $\lambda_{\mu}=5$ and $\sigma_{\mu}=10\%$, the standard deviation of the residuals becomes superior to $60\%$. 
Even with a good regime, the impact of parameters mis-specification on trend filtering is not negligible.
\begin{figure}[H]
\begin{center}
    \includegraphics[totalheight=7cm]{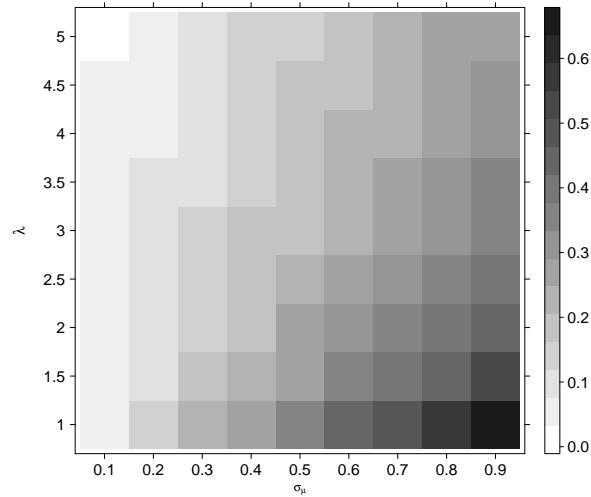}\\
    \caption{Asymptotic standard deviation of the trend as a function of the trend parameters with $\sigma_S=30\%$.}
\end{center}
\label{Figu5}
\end{figure}
\begin{figure}[H]
\begin{center}
    \includegraphics[totalheight=7cm]{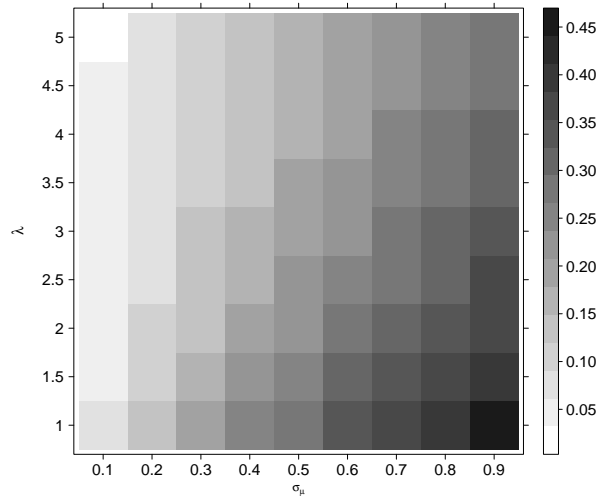}
    \caption{Asymptotic standard deviation of the residuals of the well-specified Kalman filter as a function of the trend parameters with $\sigma_S=30\%$.}
\end{center}
\label{Figu6}
\end{figure}
\begin{figure}[H]
\begin{center}
    \includegraphics[totalheight=7cm]{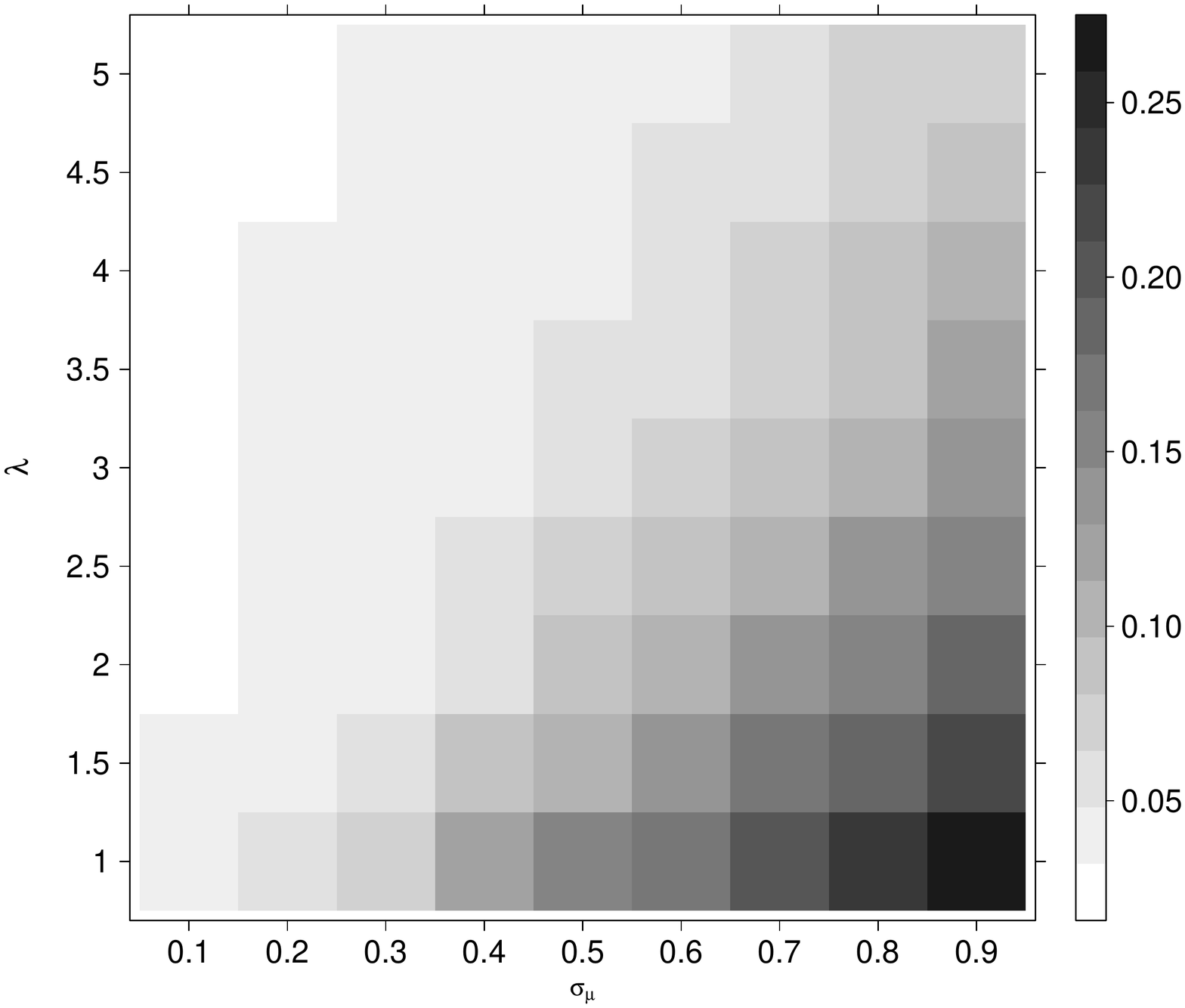}\\
    \caption{Asymptotic standard deviation of the residuals of the mis-specified Kalman filter as a function of the trend estimate parameters with $\sigma_S=30\%$, $\lambda_{\mu}^{*}=5$ and $\sigma_{\mu}^{*}=10\%$.}
\end{center}
\label{Figu7}
\end{figure}
\begin{figure}[H]
\begin{center}
    \includegraphics[totalheight=7cm]{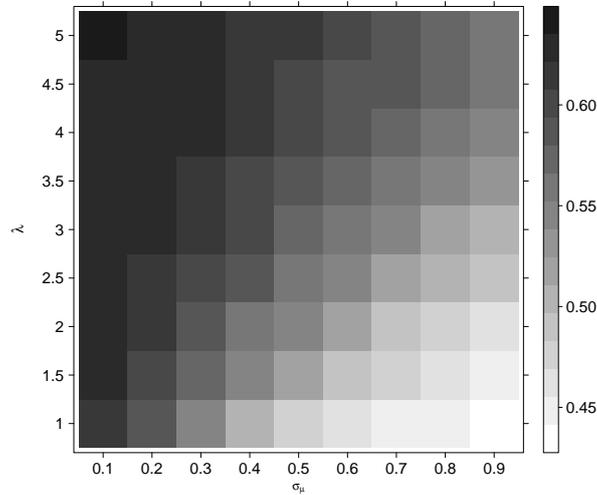}\\
    \caption{Asymptotic standard deviation of the residuals of the mis-specified Kalman filter as a function of the trend estimate parameters with $\sigma_S=30\%$, $\lambda_{\mu}^{*}=1$ and $\sigma_{\mu}^{*}=90\%$.}
\end{center}
\label{Figu8}
\end{figure}

\subsection{Detection of a positive trend}
In this subsection, using Equation (\ref{Proba}), the asymptotic probability to have a positive trend, knowing a trend estimate equal to a threshold $x$ is illustrated. In order to compare this probability for different trend regimes, we choose a threshold equal to the standard deviation of the filter $\hat{\mu}$. First, this quantity is tractable in practice. Moreover, since the continuous time mis-specified filter $\hat{\mu}$ is a centered Gaussian process, the probability that $\hat{\mu}$ becomes superior (or inferior) to its standard deviation is independent of the parameters $\left( \sigma_\mu^*,\lambda_\mu^*,\sigma_\mu,\lambda_\mu,\sigma_S \right) $ .
First, suppose that the agent uses the real values of the parameters and consider the asymptotic probability $\mathbb{P}\left(\mu^{*}>0|\hat{\mu}^{*}=\right. $ $ \left. \sqrt{\mathbb{V}_{\hat{\mu}^{*}}} \right)$ to have a positive trend, knowing an estimate equal to its standard deviation.
The figure \hyperref[Figu9]{9} represents this probability for different regimes. As seen in Proposition \ref{PropositionProba}, in the well-specified case, this probability is an increasing function of the trend volatility $\sigma_\mu^*$ and a decreasing function of $\lambda_\mu^*$. 
Like the calibration and the filtering problem, the detection is easier with a small $\lambda_{\mu}^*$ and a high drift volatility.
Now, suppose that the agent uses wrong estimates $\left(\sigma_{\mu},\lambda_{\mu} \right)$. In this case, the agent implements the continuous time mis-specified Kalman filter.
The figures \hyperref[Figu10]{10} and \hyperref[Figu11]{11} represent the asymptotic probability $\mathbb{P}\left(\mu^{*}>0|\hat{\mu}=\sqrt{\mathbb{V}_{\hat{\mu}}} \right)$ for the best and the worst configuration of the figure \hyperref[Figu9]{9}. As explained in Remark \ref{RemarkProba}, this probability is always superior to 0.5, even with a bad calibration of the parameters. For each case, the probability to have a positive trend, knowing an estimate equal to its standard deviation does not vary a lot with an error on the parameters. This quantity seems to be robust to parameters mis-specifications.
\begin{figure}[H]\label{probWellSpecified}
\begin{center}
    \includegraphics[totalheight=7cm]{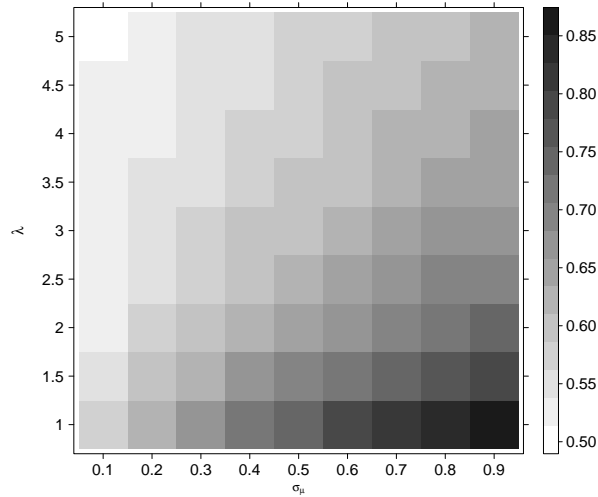}\\
    \caption{Asymptotic probability to have a positive trend given a well-specified estimate equal to its standard deviation with $\sigma_S=30\%$.}
\end{center}\label{Figu9}
\end{figure}
\begin{figure}[H]
\begin{center}
    \includegraphics[totalheight=7cm]{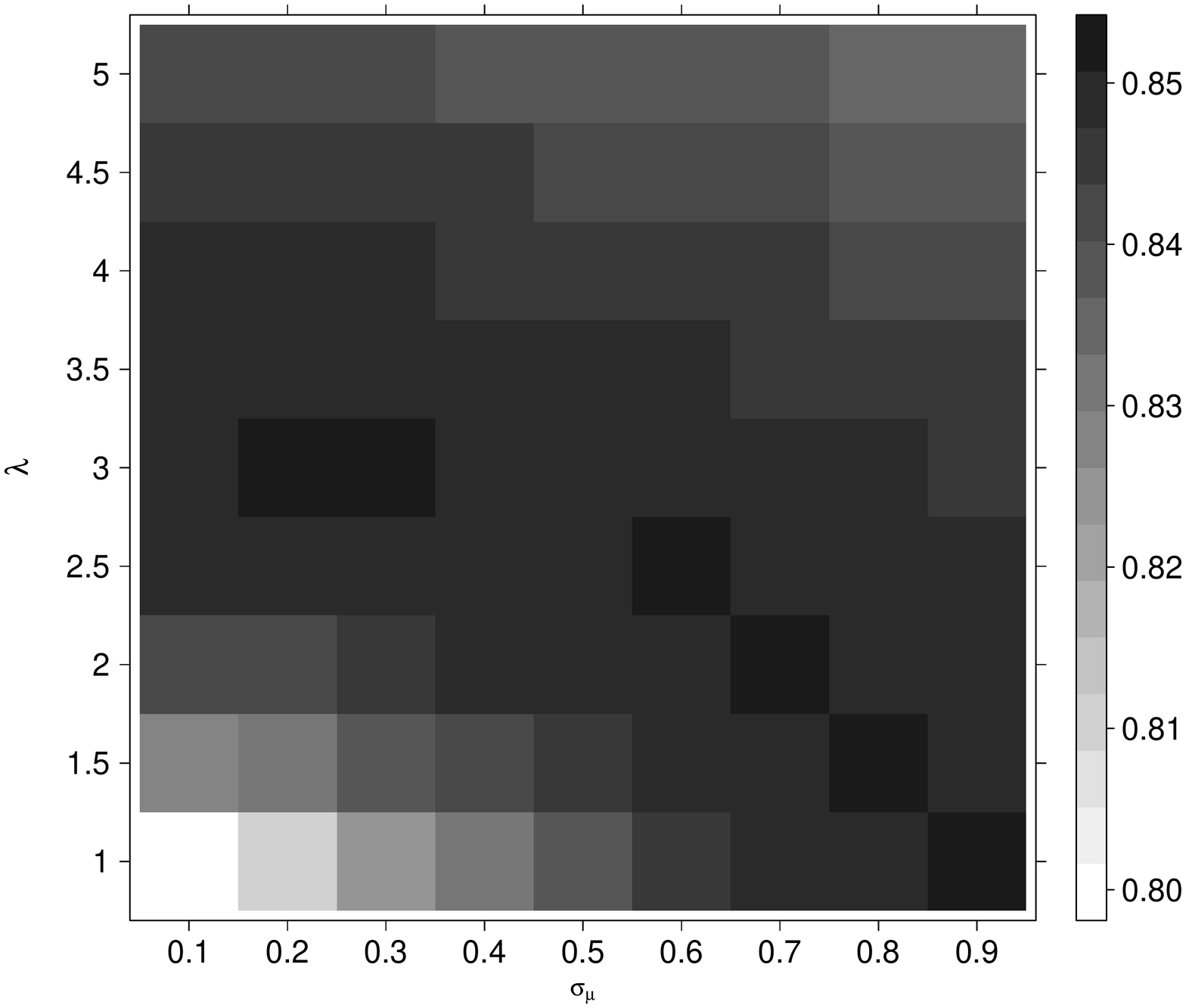}\\
    \caption{Asymptotic probability to have a positive trend given a mis-specified estimate equal to its standard deviation with $\sigma_S=30\%$, $\lambda_{\mu}^{*}=1$ and $\sigma_{\mu}^{*}=90\%$.}
\end{center}
\label{Figu10}
\end{figure}
\begin{figure}[H]
\begin{center}
    \includegraphics[totalheight=7cm]{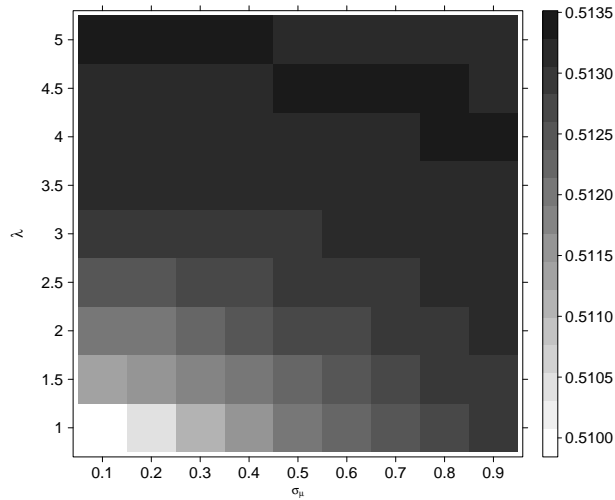}\\
    \caption{Asymptotic probability to have a positive trend given a mis-specified estimate equal to its standard deviation with $\sigma_S=30\%$, $\lambda_{\mu}^{*}=5$ and $\sigma_{\mu}^{*}=10\%$.}
\end{center}
\label{Figu11}
\end{figure}
\newpage
\section{Conclusion}
The present work tries to illustrate the difficulty of trend filtering with a model based on an unobserved mean-reverting diffusion.
This model belongs to the class of Linear Gaussian Space State models.
The advantage of this kind of system is to have an on-line method of estimation: the Kalman filter.

In practice, the parameters of the model are unknown, then the calibration of filtering parameters is crucial.
The linear and Gaussian case allows to compute, in closed form, the likelihood. The Kalman filter can also be used for this calculus. These methods can be generalized to a non-constant volatility and classical estimators can be easily put in practice. 

Although this framework is particularly convenient for forecasting, the results of the  analysis  show that the classical estimators are not adapted to a so weak signal-to-noise ratio.
Moreover, the linear and Gaussian model has deficiencies. In practice, financial asset returns are heavy-tailed (see \cite{Mandelbrot}) because of jumps, volatility fluctuations. So, the stationarity is not guaranteed on a long horizon of observations. Then, the problem of trend filtering with real data is harder than with the Linear and Gaussian framework.

With this simple model, the horizons of observations needed for a acceptable precision are too long. Therefore, the convergence is not guaranteed and the impact of mis-specification on trend filtering is not negligible. We surely conclude that it is impossible to estimate the trend with a good precision. Despite these difficulties, the non-zero correlation between the trend and the estimate (mis-specified or not) can be used for the detection of a positive (or negative) trend.

\newpage
\section*{Appendix A: discrete Kalman filter}
\label{sec::KalmanFilter}
\subsection*{Framework}
This section is based on \cite{KALMAN}. 
The discrete Kalman filter is a recursive method. Consider two objects: the observations $\left\lbrace y_{k} \right\rbrace$ and the states of the system $\left\lbrace x_{k} \right\rbrace$. This filter is based on a Gauss-Markov first order model. Consider the following system:
\begin{eqnarray*}
x_{k+1}&=&F_{k}x_{k}+v_{k},\\
y_{k}&=&H_{k}x_{k}+u_{k}.
\end{eqnarray*}
The first equation is an $a\ priori$ model, the transition equation of the system. The matrix $F_{k}$ is the transition matrix and $v_{k}$ is the transition noise. The second equation is the measurement equation. The matrix $H_{k}$ is named the measurement matrix and  $u_{k}$ is the measurement noise. 
The aim is to identify the underlying process $\left\lbrace x_{k} \right\rbrace$. The two noises are supposed white, Gaussian, centered and decorrelated. In particular:
\begin{equation*}
\mathbb{E} \left[ \left(
\begin{array}{c}
u_{k}\\
v_{k}\\
\end{array}
\right) \left(
\begin{array}{c}
u_{l}\\
v_{l}\\
\end{array}
\right)^{T}\right] = \left(
\begin{array}{cc}
R^{u}_{k}		& 0\\
0		& R^{v}_{k}\\
\end{array}
\right) \delta_{kl}.
\end{equation*}
The two noises are also supposed independent of $x_{k}$ and the initial state is Gaussian. So, it can be proved with a recurrence that all states are Gaussian. Therefore, just the mean and the covariance matrix are needed for the characterization of the state.
The estimation is given by two steps. The first one is an $a\ priori$ estimation given $\hat{x}_{k+1/k}$ and $\Gamma_{k+1/k}=\mathbb{E} \left[ (x_{k+1}-\hat{x}_{k+1/k})(x_{k+1}-\hat{x}_{k+1/k})^{T}\right]$.
When the new observation is available, a correction of the estimation is done to obtain $\hat{x}_{k+1/k+1}$ and $\Gamma_{k+1/k+1}=\mathbb{E} \left[ (x_{k+1}-\hat{x}_{k+1/k+1})(x_{k+1}- \hat{x}_{k+1/k+1})^{T}\right]$. This is the $a\ posteriori$ estimation.  The criterion considered for the $a\ posteriori$ estimation is the least squares method, which corresponds to the minimization of the trace of $\Gamma_{k+1/k+1}$.
\subsection*{Filter}
The prediction ($a\ priori$ estimation) is given by
\begin{eqnarray*}
\hat{x}_{k+1/k}&=&F_{k}\hat{x}_{k/k},\\
\Gamma_{k+1/k}&=&F_{k}\Gamma_{k/k}F_{k}^{T}+R^{v}_{k}.
\end{eqnarray*}
The $a\ posteriori$ estimation is a correction of the $a\ priori$ estimation. A gain is introduced to do this correction:
\begin{eqnarray*}
\hat{x}_{k+1/k+1}=\hat{x}_{k+1/k}+K_{k+1} \left( y_{k+1}-H_{k+1} \hat{x}_{k+1/k} \right). 
\end{eqnarray*}
As explained above, the gain $K_{k+1}$ is found by least squares method, which corresponds to
\begin{eqnarray*}
\frac{\partial trace\left( \Gamma_{k+1/k+1} \right) }{\partial K_{k+1}}=0.
\end{eqnarray*}
With the classical lemma of derivation for matrix, the gain is found:
\begin{eqnarray*}
K_{k+1}&=&\Gamma_{k+1/k} H_{k+1}^{T} \left[ H_{k+1}\Gamma_{k+1/k}H_{k+1}^{T} + R^{u}_{k+1}\right]^{-1},\\
\Gamma_{k+1/k+1}&=&\left( I_{d}-K_{k+1}H_{k+1} \right) \Gamma_{k+1/k}.
\end{eqnarray*}
\newpage
\section*{Appendix B: Iterative methods for the inverse and the determinant of the covariance matrix}
\label{sec::IterativeLikelihood}
In this appendix, we provide iterative methods for the inverse and the determinant of the covariance matrix.
\subsection*{Inverse of the covariance matrix}
The use of the Matrix Inversion Lemma on Equation (\ref{Cov}) gives:
\begin{eqnarray*}
\Sigma_{y_{1:N}|\theta}^{-1}=\Sigma_{\mu_{1:N}|\theta}^{-1}-\Sigma_{\mu_{1:N}|\theta}^{-1}A_{N}^{-1}\Sigma_{\mu_{1:N}|\theta}^{-1},
\end{eqnarray*}
where $A_N=\frac{\delta}{\sigma_{S}^{2}}I_N+\Sigma_{\mu_{1:N}|\theta}^{-1}$. Then, we have to compute the inverse of the matrices $A_N$ and $\Sigma_{\mu_{1:N}|\theta}$.
\subsubsection*{Inverse of the matrix $A_N$}
Suppose that $A_{N}^{-1}$ is computed. The matrix $A_{N+1}$ can be broken into four sub-matrices:
\begin{eqnarray*}
A_{N+1}=\left(
\begin{array}{cc}
B_1 & B_2\\
B_3 & B_4
\end{array}
\right),
\end{eqnarray*} 
where
\begin{eqnarray*} B_1 & = & \frac{\delta}{\sigma_{S}^{2}}+\frac{2\lambda_{\mu}\left(e^{\lambda_{\mu}\delta}+e^{-\lambda_{\mu}\delta} \right)}{\sigma_{\mu}^{2}\left(e^{\lambda_{\mu}\delta}-e^{-\lambda_{\mu}\delta} \right)},\\
B_2 & = & \left(
\begin{array}{cccc}
 \frac{-2\lambda_{\mu}}{\sigma_{\mu}^{2}\left(e^{\lambda_{\mu}\delta}-e^{-\lambda_{\mu}\delta} \right) } & 0 & \cdots & 0\end{array}\right),\\
B_3 & = & B_2^{T},\\
B_4 & = & A_{N}.
\end{eqnarray*} 
Therefore, the matrix $A_{N+1}$ can be inverted blockwise.
\subsubsection*{Inverse of the matrix $\Sigma_{\mu_{1:N}|\theta}$}
The following lemma is used (see \cite{OuMatriceInversion} for details):
\begin{lemme}\label{Lemma}
Let $\mu$ be an Ornstein Uhlenbeck process with parameters $\theta=\left(\lambda_{\mu},\sigma_{\mu}\right)$. The covariance matrix of $\mu_1,..,\mu_N$ is $\Sigma_{\mu_{1:N}|\theta}$. Then: 
\begin{eqnarray*}
&\scriptscriptstyle \Sigma_{\mu_{1:N}|\theta}^{-1}=\frac{2\lambda_{\mu}}{\sigma_{\mu}^{2}\left(e^{\lambda_{\mu}\delta}-e^{-\lambda_{\mu}\delta} \right) } B_N,&\\
& \scriptscriptstyle B_N=\scriptscriptstyle\left(
\scriptscriptstyle\begin{array}{cccccc}
\scriptscriptstyle e^{\lambda_{\mu}\delta}+e^{-\lambda_{\mu}\delta}&\scriptscriptstyle -1 &\scriptscriptstyle 0 &\scriptscriptstyle \cdots  &\scriptscriptstyle \cdots &\scriptscriptstyle 0\\
\scriptscriptstyle -1& \scriptscriptstyle e^{\lambda_{\mu}\delta}+e^{-\lambda_{\mu}\delta}&\scriptscriptstyle -1 &\scriptscriptstyle \ddots &  &\scriptscriptstyle \vdots\\
\scriptscriptstyle 0& \scriptscriptstyle -1&  \scriptscriptstyle e^{\lambda_{\mu}\delta}+e^{-\lambda_{\mu}\delta} &\scriptscriptstyle -1 &  &\scriptscriptstyle  \vdots\\
\scriptscriptstyle\vdots&\scriptscriptstyle \ddots &\scriptscriptstyle \ddots &\scriptscriptstyle \ddots & \scriptscriptstyle\ddots &\scriptscriptstyle \vdots\\
\scriptscriptstyle\vdots&  &\scriptscriptstyle -1 & \scriptscriptstyle e^{\lambda_{\mu}\delta}+e^{-\lambda_{\mu}\delta} & \scriptscriptstyle -1 &\scriptscriptstyle 0\\
\scriptscriptstyle \vdots&  &\scriptscriptstyle \ddots &\scriptscriptstyle -1 & \scriptscriptstyle e^{\lambda_{\mu}\delta}+e^{-\lambda_{\mu}\delta} &\scriptscriptstyle -1\\
\scriptscriptstyle 0&\scriptscriptstyle  \cdots&\scriptscriptstyle \cdots &\scriptscriptstyle 0 &\scriptscriptstyle -1 &\scriptscriptstyle e^{\lambda_{\mu}\delta}\\
\end{array}
\scriptscriptstyle\right).
\end{eqnarray*} 
\end{lemme}
Therefore, the inverse of the matrix $\Sigma_{\mu_{1:N+1}}$ is given by:
\begin{eqnarray*}
\Sigma_{\mu_{1:N+1}|\theta}^{-1}=\left(
\begin{array}{cc}
  \frac{2\lambda_{\mu}\left(e^{\lambda_{\mu}\delta}+e^{-\lambda_{\mu}\delta} \right)}{\sigma_{\mu}^{2}\left(e^{\lambda_{\mu}\delta}-e^{-\lambda_{\mu}\delta} \right)} & \begin{array}{cccc} \frac{-2\lambda_{\mu}}{\sigma_{\mu}^{2}\left(e^{\lambda_{\mu}\delta}-e^{-\lambda_{\mu}\delta} \right) } & 0 & \cdots & 0  \\ \end{array}  \\
 \begin{array}{c} \frac{-2\lambda_{\mu}}{\sigma_{\mu}^{2}\left(e^{\lambda_{\mu}\delta}-e^{-\lambda_{\mu}\delta} \right) } \\ 0 \\ \vdots \\ 0  \\ \end{array}
 &\Sigma_{\mu_{1:N}|\theta}^{-1}
\end{array}
\right).
\end{eqnarray*} 
\paragraph{Procedure}
Finally, at time $t$, the inverse of the covariance matrix is given by the following protocol:
\begin{itemize}
\item Computation of the matrix $A_t^{-1}$ using $A_{t-1}^{-1}$.
\item Computation of the matrix $\Sigma_{\mu_{1:t}|\theta}^{-1}$ using $\Sigma_{\mu_{1:t-1}|\theta}^{-1}$.
\item Using $\Sigma_{y_{1:t}|\theta}^{-1}=\Sigma_{\mu_{1:t}|\theta}^{-1}-\Sigma_{\mu_{1:t}|\theta}^{-1}A_{t}^{-1}\Sigma_{\mu_{1:t}|\theta}^{-1}$, the matrix $\Sigma_{y_{1:t}|\theta}^{-1}$ is obtained.
\end{itemize}

\subsection*{Determinant of the covariance matrix}
The iterative computation of $\det\left( \Sigma_{y_{1:N}|\theta}\right)$ is based on the following lemma:
\begin{lemme}\label{lemma}
The determinant of the matrix $ \Sigma_{y_{1:N}|\theta}$ is given by:
\begin{eqnarray}\label{Determinant}
\det\left( \Sigma_{y_{1:N}|\theta} \right) = \frac{\det \left( I_N+\frac{\sigma_{S}^{2}}{\delta} \Sigma_{\mu_{1:N}|\theta}^{-1} \right) }{\det\left( \Sigma_{\mu_{1:N}|\theta}^{-1} \right) },
\end{eqnarray}
and for $N \geq 2$, we have:
 \begin{eqnarray*}
 \scriptstyle\det\left( \Sigma_{\mu_{1:N+1}|\theta}^{-1} \right)&=&\scriptstyle\ g\left(\lambda_{\mu},\sigma_{\mu} \right) \left(e^{\lambda_{\mu}\delta}+e^{-\lambda_{\mu}\delta} \right) \det\left( \Sigma_{\mu_{1:N}|\theta}^{-1} \right)
 \\&&\scriptstyle\ -g\left(\lambda_{\mu},\sigma_{\mu} \right)^{2} \det\left( \Sigma_{\mu_{1:N-1}|\theta}^{-1} \right),\\
 \scriptstyle\det \left( I_{N+1}+\frac{\sigma_{S}^{2}}{\delta} \Sigma_{\mu_{1:N+1}|\theta}^{-1} \right) &=&\scriptstyle\ 
 \left(1+\frac{\sigma_{S}^{2}}{\delta}g\left(\lambda_{\mu},\sigma_{\mu} \right) \left(e^{\lambda_{\mu}\delta}+e^{-\lambda_{\mu}\delta} \right)  \right) \det \left( I_{N}+\frac{\sigma_{S}^{2}}{\delta} \Sigma_{\mu_{1:N}|\theta}^{-1} \right)
 \\&&\scriptstyle\ -\left(\frac{\sigma_{S}^{2}}{\delta} g\left(\lambda_{\mu},\sigma_{\mu} \right)\right) ^{2} 
 \det \left( I_{N-1}+\frac{\sigma_{S}^{2}}{\delta} \Sigma_{\mu_{1:N-1}|\theta}^{-1} \right),
 \end{eqnarray*}
 where
 \begin{eqnarray*}
 g\left(\lambda_{\mu},\sigma_{\mu} \right) =\frac{2\lambda_{\mu}}{\sigma_{\mu}^{2}\left(e^{\lambda_{\mu}\delta}-e^{-\lambda_{\mu}\delta} \right)}.
 \end{eqnarray*}
\end{lemme} 
\begin{proof}
The multiplication of Equation (\ref{Cov}) by $\Sigma_{\mu_{1:N}|\theta}^{-1}$ gives:
\begin{eqnarray*}
\Sigma_{\mu_{1:N}|\theta}^{-1} \Sigma_{y_{1:N}|\theta}&=&I_N+\frac{\sigma_{S}^{2}}{\delta} \Sigma_{\mu_{1:N}|\theta}^{-1}.
\end{eqnarray*} 
Equation (\ref{Determinant}) follows. Using Lemma \ref{Lemma}, The matrices \\$\left( I_N+\frac{\sigma_{S}^{2}}{\delta} \Sigma_{\mu_{1:N}|\theta}^{-1} \right)$ and $\Sigma_{\mu_{1:N}|\theta}^{-1}$ are tridiagonal. The recursive computation of their determinant is then possible.
\end{proof}

\newpage
\vspace{60pt}
\bibliographystyle{alpha}
\bibliography{bibli}
\end{document}